\documentclass[11pt,a4paper]{article}

\usepackage[utf8]{inputenc}
\usepackage[T1]{fontenc}
\usepackage{lmodern}
\usepackage[margin=2.5cm]{geometry}
\usepackage{microtype}
\usepackage{amsfonts}
\usepackage{amsmath,amssymb,amsthm}
\usepackage{nicefrac}
\usepackage{booktabs}
\usepackage{makecell}
\usepackage{graphicx}
\usepackage{xcolor}
\usepackage{url}
\usepackage[numbers,sort&compress]{natbib}
\usepackage{authblk}
\usepackage[hidelinks]{hyperref}

\newtheorem{theorem}{Theorem}
\newtheorem{lemma}{Lemma}
\newtheorem{proposition}{Proposition}
\newtheorem{corollary}{Corollary}
\theoremstyle{remark}
\newtheorem*{remark}{Remark}

\title{Exact Fixed-Point Constraints in Neural-ODEs: Injecting Prior Knowledge into the Network Parameters}

\author[1,2]{Feliciano Giuseppe Pacifico\thanks{E-mail: \texttt{felicianogiuseppe.pacifico@phd.unipi.it}}}
\author[2]{Duccio Fanelli}
\author[2]{Lorenzo Buffoni}
\author[2]{Lorenzo Chicchi}
\author[2]{Diego Febbe}
\author[2]{Raffaele Marino}
\affil[1]{Department of Informatics and Computer Science, University of Pisa, Italy}
\affil[2]{Department of Physics and Astronomy and INFN, University of Florence, Sesto Fiorentino, Italy}
\date{}

\begin{document}

\maketitle

\begin{abstract}
Neural ordinary differential equations (Neural-ODEs) learn a velocity field from data. In many practical settings, however, additional information on the target field is available, such as the location of some of its equilibria or, more generally, its values at given points. We introduce a method to inject this information directly into the parameters of a Neural-ODE, so that a finite collection of prescribed fixed points holds exactly throughout training, with no penalty terms and no post-hoc projection. We test the method on vector field regression for a Lotka-Volterra competition model with four equilibria. The constrained model keeps the fixed-point residual at single-precision round-off, and its test error is up to about five times lower than that of an unconstrained network of the same architecture. A soft penalty on the equilibria lowers the residual only at the cost of accuracy, and still leaves it at least an order of magnitude above that of the constrained model. When all models are trained until convergence on datasets of increasing size, the test error of the networks without hard constraints saturates already for small datasets, while that of the constrained model, with the same architecture and number of parameters, keeps decreasing over a wider range. Close to the equilibria and over long time horizons, the constrained model is the most accurate at every dataset size considered. The method relies on a parametrization of the output matrix through the Moore-Penrose pseudoinverse, implemented by a thin QR factorization of the hidden features at the prescribed points. This parametrization splits the output matrix into a component determined by the constraints and a free component, which shapes the field away from the prescribed points without altering its values there. As a result, the constrained Neural-ODE remains a universal approximator of the continuous vector fields compatible with the imposed constraints.
\end{abstract}

\section{Introduction}

Neural ordinary differential equations (Neural-ODEs) \cite{chen2018neuralode} define a class of deep learning models that aim to represent neural network layers as a continuous-time system. Instead of a discrete sequence of hidden layers, a Neural-ODE parameterizes the derivative of the state vector via a neural network. Concretely, the forward pass materializes as the evolution of a point in a multi-dimensional space through a system of coupled ordinary differential equations, written in terms of a deep neural network; the input-output maps of Neural-ODEs are universal approximators within suitable classes of invertible maps \cite{teshima2020universal, ishikawa2023universal}. Since the network that defines the vector field is itself a universal approximator \cite{leshno1993multilayer}, a Neural-ODE can also serve, beyond its use as a continuous-depth architecture, as a proxy for a dynamical system whose governing equations are not known in closed form, but for which data are available. The learned field then stands in for the unknown one.

In many practical settings, however, data are not the only source of information. Knowledge of the underlying system often provides properties that the vector field must satisfy, such as the location of some of its fixed points or, more generally, the values that the field must take at given points. The standard approach trains the Neural-ODE on the data alone, in the hope that this additional information is recovered during training, or adds to the loss a soft penalty, which enforces it only approximately. In this work we propose instead to inject this information directly into the parameters of the network that serves as a proxy for the vector field. The resulting model satisfies the prescribed constraints exactly at every step of training. In the experiments of Section~\ref{sec:exp_vectorfield_regression}, it is also more accurate than networks of the same architecture trained on the data alone or with a soft penalty, not only at the prescribed points but over the whole domain.

Injecting this information into the network amounts to imposing hard constraints on it. These are conditions that the model satisfies exactly by construction, rather than conditions that the training loss merely encourages. Hard-constrained neural architectures have been investigated from several complementary viewpoints. Exact interpolation of finitely many prescribed values by feed-forward networks has been studied together with their approximation capabilities \cite{CAO2010457,constantinescu2023approximation}, while boundary or algebraic conditions have been imposed exactly on neural-network solutions of differential equations, from the trial-solution ansatz of Lagaris et al.~\cite{lagaris1998artificial} to recent physics-informed neural networks based on distance functions \cite{sukumar2022exact} or on hard linear equality constraints \cite{chen2024physics}. Input-dependent hard constraints, prescribed at every input, can also be enforced through a differentiable projection of the network output without losing universal approximation guarantees \cite{min2024hardnet}. For continuous-depth models, related ideas have been explored in several directions, including manifold-constrained Neural-ODEs with universality results \cite{elamvazhuthi2023learning} and stability-oriented Neural-ODE constructions based on Lyapunov-stable equilibria \cite{kang2021stable}. Exact feasibility has also been enforced via differentiable optimization layers \cite{amos2017optnet}. More recently, dedicated methods have been proposed to constrain the learned dynamics \cite{GaglianiPacificoChicchiFanelliFebbeBuffoniMarino2026, Chicchi_2025} via spectral decomposition and orthogonal projection operators. Other approaches project the vector field onto the tangent space of a constraint manifold \cite{white2024projected}. These works show that structural priors and exact constraints can be built into trainable architectures.

Yet, the above results do not provide a parametrization of a Neural-ODE velocity field that satisfies finitely many prescribed fixed points exactly, separates the parameters determined by the constraints from those left free for training, and retains universality within the class of compatible target dynamics. The closest work is the general analysis of Constantinescu and Popescu \cite{constantinescu2023approximation}. Under minimal assumptions on the activation, they proved that networks of arbitrary depth can interpolate any finite data set and are dense in the space of continuous functions. They also showed that the interpolating parameters generically form a smooth manifold, and gave a probabilistic procedure that computes an interpolating output layer by linear regression on random hidden features. We share their interpolation principle, but pursue a different goal. In \cite{constantinescu2023approximation} the interpolated points are the whole data set, which the network is required to fit exactly. Here, instead, only a few values are prescribed, encoding prior knowledge on the target field; the remaining data are fitted by regression, and the constraints serve to steer the model towards the true vector field rather than to reproduce the samples. Moreover, their probabilistic procedure draws the hidden layer at random and computes a single output layer, whereas the goal of this work is to train a model on the data by separating the parameters needed to enforce the constraints from the free ones used to approximate the vector field away from the prescribed point constraints. This description is available whenever the feature matrix at the prescribed points has full column rank, and keeps the constraints exact throughout training, hidden parameters included. In our setting the approximated function is the velocity field of a Neural-ODE, which acts as a proxy for an unknown dynamical system; the prescribed values then become prescribed equilibria, and the separation between constrained and free parameters allows the prior knowledge to be retained while the field is learned. HardNet \cite{min2024hardnet} also keeps hard constraints exact during end-to-end training, but it does so by projecting the network output, input by input, onto a feasible set prescribed at every input. With constraints located at isolated points, this projection either restricts the field on whole neighborhoods of those points, so that the admissible class is strictly smaller than the class of fields vanishing there, or only overwrites its values at the points themselves. In the latter case the field becomes discontinuous, and its prescribed zeros have no effect on the nearby dynamics. The resulting architecture might, at least in principle, lose expressive power precisely because of the imposed equilibrium constraints.

The goal of this work is to bridge the above viewpoints by injecting the information carried by the constraints into the vector field itself, through its parameters. For fixed hidden parameters, the fixed-point conditions are linear in the output matrix. We parametrize the complete set of admissible output matrices and implement it through a thin QR factorization of the hidden features evaluated at the $C$ prescribed points only. The constraints then hold exactly for every value of the free parameters, at every iterate of a gradient-based training that updates all the parameters, hidden ones included, while the constrained model remains a continuous network on the whole domain. Finite interpolation by neural networks and the principle of approximating first and then correcting finitely many values are both well established. Theorem~\ref{thm:main} builds on them to support the construction: besides exact constraints and uniform approximation, it provides an approximant whose feature matrix has full column rank, i.e.\ one that lies in the domain of the parametrization. As a consequence, the trainable constrained family, including the residual Neural-ODE used in the experiments, is dense in the class of continuous vector fields sharing the prescribed fixed points. A technical remark on the relation with \cite{constantinescu2023approximation} is given in Section~\ref{subsec:pinv_planting}.

The paper is organized as follows. Section \ref{sec:2} defines the problem, while Section \ref{theorem} presents the main theorem, which establishes universal approximation under finitely many exact pointwise constraints for any continuous nonpolynomial activation, together with its proof. In Section \ref{sec:planting_node} we turn to the procedure that enforces the fixed points in the generated velocity fields: we show that the pseudoinverse parametrization, and its equivalent QR implementation, represents every admissible output matrix, so that the constrained models inherit the universality of the theorem. Vector field regression for models with planted fixed points is demonstrated in Section \ref{sec:exp_vectorfield_regression}, where the constrained model is also compared with unconstrained and penalized MLPs.

\section{Stating the problem}\label{sec:2}

Consider a Neural-ODE model that we generally formulate as:
\begin{equation}
\dot x=A \sigma(Bx+b) \equiv F_\theta(x)
\end{equation}
where $A\in\mathbb R^{n\times m}$, $B\in\mathbb R^{m\times n}$, $b\in\mathbb R^m$, and $\sigma:\mathbb R\to\mathbb R$ is a continuous nonlinear activation function, applied component-wise to vectors. Throughout, $\sigma$ is only assumed to be continuous and not a polynomial; this covers, among others, ReLU, GELU \cite{hendrycks2016gelu}, sigmoid and hyperbolic tangent. The multi-layer perceptron that defines the velocity field of the above Neural-ODE is denoted by $F_\theta(x)$ in short, where $\theta$ stands for the cumulative set of free tunable parameters. Imagine that a procedure exists to insert in the above vector field
$F_\theta(x)$ a set of $C$ fixed-point constraints, namely:
\begin{equation}
\label{eq:fp_constraints_exp}
F_\theta(\bar x^{(l)}) = 0,\qquad l=1,\dots,C
\end{equation}

More generally, one may prescribe the values of the vector field at the chosen points, $F_\theta(\bar x^{(l)})=c^{(l)}$ with given $c^{(l)}\in\mathbb R^n$; fixed points correspond to $c^{(l)}=0$. The specific procedure that we have devised to enforce these constraints will be illustrated in Section \ref{sec:planting_node}.

Given the above, a natural question arises: does imposing hard equilibrium constraints reduce the expressive power of the Neural-ODE vector field? This is the question that we set out to answer in the next Section. More specifically, we will prove that, on compact sets, the family of vector fields represented by a one-hidden-layer Neural-ODE remains a universal approximator, even when restricted to the subclass satisfying the prescribed constraints. In other words, the constrained parametrizations enforce a local condition on the vector field without sacrificing universality within the class of target dynamics that share the same constraints.

\section{Universality under exact pointwise constraints}
\label{theorem}

As discussed in the Introduction, in this work the hard constraints are used to prescribe fixed points of the Neural-ODE vector field. The theorem below shows that the constrained family remains dense in the class of continuous vector fields taking the prescribed values at the prescribed points; hence finitely many exact pointwise constraints, and in particular exact equilibrium constraints, do not reduce the approximation power of the model class.

\begin{figure}[t]
    \centering
    \includegraphics[width=0.55\linewidth]{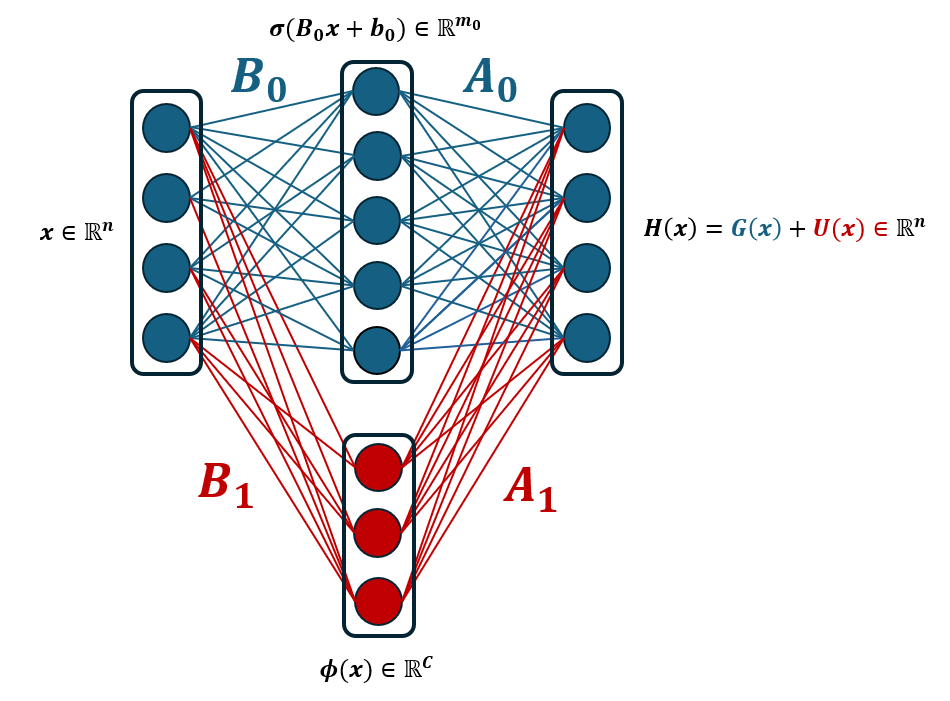}
    \caption{A schematic layout of the architectural construction employed in the proof of Theorem~\ref{thm:main}. Blue: the unconstrained approximant $G(x)=A_0\sigma(B_0x+b_0)$, with $m_0$ hidden units. Red: the $C$ correction neurons $\phi(x)=\sigma(B_1x+b_1)$ and their output block $A_1$, which produce the correction $U(x)$.}
    \label{scheme}
\end{figure}

\begin{theorem}[Uniform approximation of vector fields with prescribed pointwise values]
\label{thm:main}
Let $n\ge 1$, let $K\subset\mathbb R^n$ be compact, and let $\bar x^{(1)},\dots,\bar x^{(C)}\in K$ be $C\ge 1$ distinct points. Let $\sigma:\mathbb R\to\mathbb R$ be continuous and not coincide with any algebraic polynomial on all of $\mathbb R$. Fix prescribed values $c^{(1)},\dots,c^{(C)}\in\mathbb R^n$, and collect them in $Y=[c^{(1)},\dots,c^{(C)}]\in\mathbb R^{n\times C}$. For $m\ge1$, consider the class of network vector fields
\[
\mathcal N_m(\sigma)=\big\{H(x)=A\sigma(Bx+b):\ A\in\mathbb R^{n\times m},\ B\in\mathbb R^{m\times n},\ b\in\mathbb R^m\big\},
\]
with $\sigma$ applied componentwise, its constrained subclass
\[
\mathcal N_{m,Y}(\sigma)=\big\{H\in\mathcal N_m(\sigma):\ H(\bar x^{(l)})=c^{(l)}\ \text{for all } l\big\},
\]
and the constrained target class
\[
\mathcal F_Y=\big\{F\in C(K,\mathbb R^n):\ F(\bar x^{(l)})=c^{(l)}\ \text{for all } l\big\}.
\]
Then, for every $F\in\mathcal F_Y$ and every $\varepsilon>0$, there exist a finite hidden width $m\ge C$ and a network $H\in\mathcal N_{m,Y}(\sigma)$ such that
\begin{equation}
\label{eq:thm_claims}
\sup_{x\in K}\|F(x)-H(x)\|_2<\varepsilon,
\end{equation}
where $\|\cdot\|_2$ denotes the Euclidean norm in $\mathbb R^n$. Moreover, the parameters of $H$ can be chosen so that the feature matrix
\begin{equation}
\label{eq:feature_matrix}
S=\big[\sigma(B\bar x^{(1)}+b),\dots,\sigma(B\bar x^{(C)}+b)\big]\in\mathbb R^{m\times C}
\end{equation}
has rank $C$.
\end{theorem}

\begin{proof}
Fix $F$ and $\varepsilon$ as in the statement. By Lemma~\ref{lem:invertible_M}, proved in Appendix~\ref{app:lemma}, we may choose $C$ parameter pairs $(w_k,\beta_k)$ for which the evaluation matrix $M$, with entries $M_{lk}=\sigma(w_k^\top\bar x^{(l)}+\beta_k)$, is invertible. Using these fixed neurons, let $B_1\in\mathbb R^{C\times n}$ have rows $w_k^\top$, let $b_1=(\beta_1,\dots,\beta_C)^\top$, and define
\[
\phi(x)=\sigma(B_1x+b_1)=
\begin{bmatrix}
\sigma(w_1^\top x+\beta_1)\\
\vdots\\
\sigma(w_C^\top x+\beta_C)
\end{bmatrix},
\qquad
\ell(x)=(M^\top)^{-1}\phi(x).
\]
Write $\ell(x)=(\ell_1(x),\dots,\ell_C(x))^\top$, so that each $\ell_i$ is a scalar function. Let $e_l$ denote the $l$-th standard basis vector of $\mathbb R^C$. By the definition of $M$,
\begin{equation}
\label{eq:cardinal}
\phi(\bar x^{(l)})=M^\top e_l,\qquad \ell(\bar x^{(l)})=e_l.
\end{equation}
In particular, $\ell_i(\bar x^{(l)})=1$ if $i=l$, and is zero otherwise.

Each affine map $x\mapsto w_k^\top x+\beta_k$ is continuous, so each component of $\phi$ is continuous. Multiplication by the fixed matrix $(M^\top)^{-1}$ preserves continuity. Therefore $\ell(x)$ is continuous and compactness of $K$ gives
\begin{equation}
\label{eq:LambdaK}
\Lambda_K:=\max_{x\in K}\sum_{l=1}^C|\ell_l(x)|<\infty.
\end{equation}
No global bound on $\sigma$ is required. The correction features, $M$, and $\Lambda_K$ are fixed before choosing an unconstrained approximant.

Set $\delta=\varepsilon/(1+\Lambda_K)$. The unconstrained universal approximation theorem of Leshno, Lin, Pinkus, and Schocken \cite{leshno1993multilayer} states that, for every compact $L\subset\mathbb R^n$,
\begin{equation}
\label{eq:leshno}
\operatorname{span}\{\sigma(w^\top x+\beta):w\in\mathbb R^n,\ \beta\in\mathbb R\}\quad\text{is dense in } C(L,\mathbb R)
\end{equation}
in the uniform norm. Applying it to each component of $F$ and stacking the resulting hidden units gives a width $m_0$ and parameters $A_0,B_0,b_0$ such that
\[
G(x)=A_0\sigma(B_0x+b_0),\qquad \sup_{x\in K}\|F(x)-G(x)\|_2<\delta.
\]
Define residual vectors and their matrix by
\begin{equation}
\label{eq:residuals}
r_l=c^{(l)}-G(\bar x^{(l)}),\qquad R=[r_1,\dots,r_C]\in\mathbb R^{n\times C}.
\end{equation}
Because $F(\bar x^{(l)})=c^{(l)}$ and $\bar x^{(l)}\in K$, each residual satisfies $\|r_l\|_2<\delta$.

We correct these residuals by defining
\[
A_1=R(M^\top)^{-1},\qquad U(x)=A_1\phi(x)=R\,\ell(x),\qquad H(x)=G(x)+U(x).
\]
Equation \eqref{eq:cardinal} gives
\[
U(\bar x^{(l)})=Re_l=r_l,\qquad H(\bar x^{(l)})=G(\bar x^{(l)})+r_l=c^{(l)}.
\]

The corrected field is still a network with one hidden layer (see Figure~\ref{scheme}). Set
\begin{equation}
\label{eq:concat}
A=\begin{bmatrix}A_0 & A_1\end{bmatrix},\qquad
B=\begin{bmatrix}B_0\\ B_1\end{bmatrix},\qquad
b=\begin{bmatrix}b_0\\ b_1\end{bmatrix}.
\end{equation}
Then $H(x)=A\sigma(Bx+b)$ and $m=m_0+C$. For this construction, the feature matrix is
\[
S=\begin{bmatrix}S_0\\ M^\top\end{bmatrix},\qquad
S_0=\big[\sigma(B_0\bar x^{(1)}+b_0),\dots,\sigma(B_0\bar x^{(C)}+b_0)\big].
\]
Since $S$ has $C$ columns and contains the invertible $C\times C$ block $M^\top$, $\operatorname{rank}(S)=C$.

It remains to bound the error introduced by the correction. For each $x\in K$, the triangle inequality gives
\[
\|U(x)\|_2=\Big\|\sum_{l=1}^C r_l\,\ell_l(x)\Big\|_2
\le\sum_{l=1}^C\|r_l\|_2\,|\ell_l(x)|
\le\delta\Lambda_K.
\]
Consequently,
\[
\sup_{x\in K}\|F(x)-H(x)\|_2\le\sup_{x\in K}\|F(x)-G(x)\|_2+\sup_{x\in K}\|U(x)\|_2<\delta(1+\Lambda_K)=\varepsilon.
\]
This proves the approximation estimate, together with the exact constraints and the full-column-rank property.
\end{proof}

For $Y=0$, the prescribed points are equilibria of the Neural-ODE $\dot x=H(x)$, and $\mathcal F_Y$ reduces to the class $\mathcal F_0$ of continuous vector fields vanishing at them. The general case is needed nonetheless: the residual model used in the experiments corresponds to the nonzero values $c^{(l)}=\bar x^{(l)}-d$ (Corollary~\ref{cor:residual}). Finally, the width $m$ depends on $F$ and $\varepsilon$: the theorem concerns the constrained networks of all hidden widths, and does not assert universality at a fixed finite width.

\section{Planting fixed points in Neural-ODEs}
\label{sec:planting_node}

The universality theorem proved above is an \emph{existence} result. Stated differently, it guarantees that the constrained networks of arbitrary width are dense in the class $\mathcal F_Y$ of compatible vector fields.
In this Section, we provide an operative perspective, by elaborating on a \emph{constructive and training-friendly} recipe to impose the constraints
\(
F_\theta(\bar x^{(l)})=c^{(l)}
\),
and in particular the fixed points $F_\theta(\bar x^{(l)})=0$,
within a Neural-ODE, so that gradient-based training cannot drift outside the constrained class (namely the class of fields that satisfy the prescribed constraints). We then show that this recipe does not lose the universality guaranteed by Theorem~\ref{thm:main}.

\subsection{Pseudoinverse-based planting}
\label{subsec:pinv_planting}

For fixed hidden parameters $B,b$, let $S$ be given by \eqref{eq:feature_matrix}, and let $Y=[c^{(1)},\dots,c^{(C)}]$ be as in Theorem~\ref{thm:main}. The pointwise conditions for $H(x)=A\sigma(Bx+b)$ are exactly
\begin{equation}
\label{eq:AS_eq_Y}
AS=Y.
\end{equation}
As a key observation, for fixed $(B,b)$ these conditions are \emph{linear} in the output matrix $A$.
Assume that the planted features are linearly independent, i.e.
\begin{equation}
\mathrm{rank}(S)=C,
\qquad \text{which requires } m\ge C.
\label{eq:rank_assumption}
\end{equation}
Under \eqref{eq:rank_assumption}, the linear system \eqref{eq:AS_eq_Y} is solvable; its solutions form an affine space of dimension $n(m-C)$, which is nontrivial exactly when $m>C$.
Let $S^+$ denote the Moore-Penrose pseudoinverse \cite{GolubVanLoan2013,BenIsraelGreville2003}. If $\operatorname{rank}(S)=C$, then $S^+S=I_C$, and
\begin{equation}
P := SS^+ \in \mathbb R^{m\times m}
\label{eq:def_P_pinv}
\end{equation}
is the orthogonal projector in $\mathbb R^m$ onto the column space of $S$.
A canonical choice is the minimum-norm solution
\begin{equation}
A_{\mathrm{part}} := Y S^+ \qquad \Longrightarrow \qquad A_{\mathrm{part}} S = Y.
\label{eq:A_part_pinv}
\end{equation}
To preserve full freedom (for training purposes) while adjusting for the sought constraints, we proceed as follows.
Since each column of $S$ lies in the column space of $S$, we have
\begin{equation}
(I_m-P)S=0.
\label{eq:IminusP_S_zero}
\end{equation}
Therefore, for any free matrix $W\in\mathbb R^{n\times m}$ we can define
\begin{equation}
A(W) := A_{\mathrm{part}} + W(I_m-P) = YS^+ + W(I_m-SS^+),
\label{eq:A1_pinv_param}
\end{equation}
and automatically obtain $A(W)S=Y$. Hence \eqref{eq:AS_eq_Y} holds exactly, for \emph{all values of the free parameters} $W$, and in particular at every step of training. The next result shows that, conversely, no admissible output matrix is lost.

\begin{proposition}[Complete parametrization of the admissible output matrices]
\label{prop:complete}
Let $S\in\mathbb R^{m\times C}$ have rank $C$, and let $Y\in\mathbb R^{n\times C}$. Then
\begin{equation}
\label{eq:complete}
\{A\in\mathbb R^{n\times m}:\ AS=Y\}=\{YS^++W(I_m-SS^+):\ W\in\mathbb R^{n\times m}\}.
\end{equation}
In particular, every admissible output matrix is represented by the pseudoinverse parametrization.
\end{proposition}

\begin{proof}
We prove \eqref{eq:complete} by establishing the two inclusions separately.

\emph{Inclusion $\supseteq$: every matrix produced by the parametrization satisfies the constraints.}
For any $W\in\mathbb R^{n\times m}$, define
\[
A(W)=YS^++W(I_m-SS^+).
\]
Using $S^+S=I_C$, we obtain
\[
A(W)S=Y(S^+S)+W(S-SS^+S)=Y.
\]
Thus every parametrized matrix is admissible.

\emph{Inclusion $\subseteq$: every matrix that satisfies the constraints is obtained from the parametrization for some $W$.}
Conversely, let $A_*$ satisfy $A_*S=Y$. Choosing $W=A_*$ yields
\[
\begin{aligned}
A(W)&=YS^++A_*(I_m-SS^+)\\
&=YS^++A_*-A_*SS^+\\
&=YS^++A_*-YS^+=A_*.
\end{aligned}
\]
This proves the inclusion $\subseteq$ and hence \eqref{eq:complete}.
\end{proof}

Combined with Theorem~\ref{thm:main}, Proposition~\ref{prop:complete} shows that the parametrization does not lose universality. Given $F\in\mathcal F_Y$ and $\varepsilon>0$, the network $A_*\sigma(B_*x+b_*)$ provided by the theorem satisfies $A_*S_*=Y$ with $\operatorname{rank}(S_*)=C$; choosing $B=B_*$, $b=b_*$ and $W=A_*$ in \eqref{eq:A1_pinv_param} gives back $A_*$, and hence the same network. Therefore, the family of networks ${\big(YS^++W(I_m-SS^+)\big)\sigma(Bx+b)}$, with variable width and with $W,B,b$ free subject to $\operatorname{rank}(S)=C$, is dense in $\mathcal F_Y$, and each of its members satisfies the constraints exactly. Training need not reproduce the construction used in the proof of the theorem.

The same formula can be written as
\begin{equation}
\label{eq:residual_correction}
A(W)=W+(Y-WS)S^+.
\end{equation}
It corrects the constraint residual of a free matrix $W$; when $W$ already satisfies the constraints, this correction vanishes. The term $W(I_m-SS^+)$ is essential to the complete-parametrization argument. Restricting to the minimum-norm solution $YS^+$, i.e.\ to $W=0$, selects a single output matrix for fixed hidden parameters; the resulting family is not covered by the density argument above, because the output matrix constructed in Theorem~\ref{thm:main} need not be the minimum-norm one.

\begin{remark}[Relation to finite interpolation]
The minimum-norm choice $W=0$ is the output layer computed by the probabilistic procedure of \cite[Sec.~5]{constantinescu2023approximation}. While \cite{constantinescu2023approximation} shows that the set of interpolating parameters is generically a smooth manifold, Proposition~\ref{prop:complete} describes it explicitly, for one hidden layer and in the full-rank region: for given hidden parameters, the admissible output matrices form an affine family, whose free component $W(I_m-SS^+)$ is the one trained on the regression loss. Finally, the QR implementation of Section~\ref{subsec:qr_planting} avoids the inverse of $S^\top S$, whose condition number is the square of that of $S$.
\end{remark}

\begin{remark}[Computational issue]
While \eqref{eq:A1_pinv_param} is conceptually clean, computing $S^+$ via an SVD-based pseudoinverse can be
computationally expensive and hide numerical pitfalls when $S$ is ill-conditioned (especially as $C$ grows).
In Section~\ref{subsec:qr_planting} we show how to obtain the \emph{same} constrained matrix via a QR factorization, which replaces the
pseudoinverse by triangular solves.
\end{remark}

\subsection{QR-based planting}
\label{subsec:qr_planting}

Assume again \eqref{eq:rank_assumption}. Consider a thin QR decomposition \cite{GolubVanLoan2013, BenIsraelGreville2003} of the feature matrix:
\begin{equation}
S = QR,
\label{eq:qr_S}
\end{equation}
where $Q\in\mathbb R^{m\times C}$ has orthonormal columns ($Q^\top Q=I_C$) and
$R\in\mathbb R^{C\times C}$ is upper triangular and invertible.
Since the column spaces of $S$ and $Q$ coincide, the orthogonal projector onto the column space of $S$ is
\begin{equation}
P = QQ^\top,
\label{eq:P_QQt}
\end{equation}
and again $(I_m-P)S=0$ because each column of $S$ lies in the column space of $Q$.

We now find the particular solution via QR.
We want $A_{\mathrm{part}}S=Y$. Using $S=QR$, this is
\[
A_{\mathrm{part}}QR = Y.
\]
Define $Z:=A_{\mathrm{part}}Q\in\mathbb R^{n\times C}$. Then the equation becomes
\begin{equation}
ZR = Y.
\label{eq:ZR_eq_Y}
\end{equation}
Because $R$ is upper triangular and invertible, \eqref{eq:ZR_eq_Y} has the unique solution $Z=YR^{-1}$, hence
\begin{equation}
A_{\mathrm{part}} = ZQ^\top = (YR^{-1})Q^\top.
\label{eq:A_part_qr}
\end{equation}
Indeed, $S^+=R^{-1}Q^\top$ and $SS^+=QQ^\top$, so that this is the same minimum-norm solution as in \eqref{eq:A_part_pinv}.
In the numerical implementations it is not needed to form $R^{-1}$ explicitly: one solves the triangular system
$ZR=Y$ by back-substitution (an $O(C^2)$ operation per output dimension).

By adding now the free nullspace component, as in the pseudoinverse case, we obtain
\begin{equation}
A(W) = A_{\mathrm{part}} + W(I_m-P)
=
YR^{-1}Q^\top + W(I_m-QQ^\top),
\qquad W\in\mathbb R^{n\times m}\ \text{free},
\label{eq:A1_qr_param}
\end{equation}
which represents the same family as \eqref{eq:A1_pinv_param}. In particular, Proposition~\ref{prop:complete} and the density statement that follows it apply verbatim to the QR parametrization. The free component can be evaluated as $W-(WQ)Q^\top$, without forming the $m\times m$ projector. If the hidden parameters change, $S$ changes and the factorization must in general be recomputed; in practice this is done once per forward pass.

\medskip
In conclusion, the pseudoinverse and QR constructions enforce the \emph{same} constraints and parametrize the same family. The QR version should be preferred in practice because it avoids SVD-based pseudoinverses and the normal equations, and replaces
matrix inversion by triangular solves on $R$ \cite{GolubVanLoan2013,BenIsraelGreville2003}. Note, however, that $R$ has the same singular values as $S$, so that the QR route does not cure an ill-conditioned feature matrix.
Moreover, the decomposition
\(
A(W) = A_{\mathrm{part}} + W(I_m-P)
\)
makes explicit that the constraints fix only the action of the output matrix on the column space of $S$, while leaving the
orthogonal complement free for learning task-specific dynamics; hence, throughout training, the model is updated within the constrained class, while retaining substantial expressive freedom. If $\mathrm{rank}(S)<C$ at some iterate, the constraints \eqref{eq:AS_eq_Y} may become incompatible (or redundant).
In practice, one can (i) choose $m\ge C$, (ii) initialize so that $S$ has full column rank, and (iii) add a mild
conditioning regularizer that discourages collapse of the columns of $S$.

\subsection{Residual Neural-ODE parametrization}
\label{subsec:node_setup}

We hereafter consider a slightly modified version of the Neural-ODE velocity field as discussed above. More specifically, we set:
\begin{equation}
\dot x(t) = F_\theta(x(t)),
\qquad
F_\theta(x) := -x + A \sigma(B x + b) + d,
\label{eq:node_model}
\end{equation}
where $x(t)\in\mathbb R^n$, $B\in\mathbb R^{m\times n}$, $b\in\mathbb R^m$,
$A\in\mathbb R^{n\times m}$, $d\in\mathbb R^n$, and $\sigma$ is applied component-wise. Here $A$, $B$, $b$ play the same role as in the network $A\sigma(Bx+b)$ of Section~\ref{subsec:pinv_planting}, while $d$ is an output bias.
We used this modified model as it can be mapped, in some suitable limit, to a coupled multi-population dynamical system, in the spirit of the Funahashi-Nakamura construction \cite{FUNAHASHI1993801}. In a companion work \cite{PacificoFanelliBuffoniChicchiFebbeMarino2026}, the same model, with its equilibria planted through the parametrization derived below, is used as a dynamical classifier. Each class is represented by a planted fixed point, and training shapes the velocity field so that every input, taken as initial condition, is driven toward the equilibrium of its class; the linear term $-x$ favours the relaxation of the trajectories toward the planted points. The mapping of the model onto two interacting populations, a visible and a hidden one evolving on distinct time scales, is also worked out there. As we plan to use this constrained Neural-ODE model in future works to investigate some biologically inspired models rooted in neuroscience, we opted to investigate this slight modification without loss of generality.

Notice that in the right-hand side of \eqref{eq:node_model} we have subtracted a linear term ($-x$) without loss of generality. For the sake of completeness, we now apply biases on both layers, via $b$ (hidden layer) and $d$ (output layer). Given $C$ desired equilibria, $\{\bar x^{(l)}\}_{l=1}^C\subset\mathbb R^n$,
the fixed-point constraints are
\begin{equation}
F_\theta(\bar x^{(l)})=0,
\qquad l=1,\dots,C.
\label{eq:fp_constraints}
\end{equation}

Evaluating \eqref{eq:node_model} at $\bar x^{(l)}$ gives
\[
0
=
-\bar x^{(l)} + A \sigma(B \bar x^{(l)} + b) + d
\quad\Longleftrightarrow\quad
A s_l = \bar x^{(l)} - d,
\qquad s_l := \sigma(B \bar x^{(l)} + b)\in\mathbb R^m.
\]
For any fixed $d\in\mathbb R^n$, define
\begin{equation}
S=\big[\sigma(B\bar x^{(1)}+b),\dots,\sigma(B\bar x^{(C)}+b)\big],
\qquad
Y_d=\big[\bar x^{(1)}-d,\dots,\bar x^{(C)}-d\big].
\label{eq:def_S_Y}
\end{equation}
The equilibrium conditions \eqref{eq:fp_constraints} are equivalent to the single matrix equation
\begin{equation}
AS=Y_d,
\label{eq:A1S_eq_Y}
\end{equation}
which is of the form \eqref{eq:AS_eq_Y} with prescribed values $c^{(l)}=\bar x^{(l)}-d$. When $\operatorname{rank}(S)=C$, the full parametrization is therefore
\begin{equation}
\label{eq:A1_residual_param}
A=Y_dS^++W(I_m-SS^+)=Y_dR^{-1}Q^\top+W(I_m-QQ^\top),
\end{equation}
where $S=QR$ is a thin QR factorization as in Section~\ref{subsec:qr_planting}.

\begin{corollary}[Universality for the residual model with prescribed equilibria]
\label{cor:residual}
Let $\sigma$, $K$, and the prescribed points satisfy the assumptions of Theorem~\ref{thm:main}. For any fixed $d$, the family \eqref{eq:node_model} with $A$ given by \eqref{eq:A1_residual_param}, variable width, and $\operatorname{rank}(S)=C$, is dense in
\[
\mathcal F_0=\big\{F\in C(K,\mathbb R^n):\ F(\bar x^{(l)})=0\ \text{for all } l\big\}.
\]
\end{corollary}

\begin{proof}
Given $F\in\mathcal F_0$, define the shifted target
\[
\widetilde F(x)=F(x)+x-d.
\]
It is continuous and satisfies $\widetilde F(\bar x^{(l)})=\bar x^{(l)}-d$. Theorem~\ref{thm:main}, applied to $\widetilde F$ with prescribed values $c^{(l)}=\bar x^{(l)}-d$, gives a network $\widetilde H(x)=A\sigma(Bx+b)$ with $\operatorname{rank}(S)=C$ and
\[
AS=Y_d,\qquad \sup_{x\in K}\big\|\widetilde F(x)-\widetilde H(x)\big\|_2<\varepsilon,
\]
and Proposition~\ref{prop:complete} shows that $A$ is of the form \eqref{eq:A1_residual_param}, with $W=A$. Then
\[
H(x)=-x+\widetilde H(x)+d
\]
has all prescribed equilibria and
\[
\sup_{x\in K}\|F(x)-H(x)\|_2=\sup_{x\in K}\big\|\widetilde F(x)-\widetilde H(x)\big\|_2<\varepsilon. \qedhere
\]
\end{proof}

The result holds for every fixed $d$; allowing $d$ to be trainable also includes these families, with $Y_d$ evaluated at its current value. Hence the universality result carries over to the residual model, in the very parametrization used in the numerical experiments.

\section{Numerical experiments: vector field regression with planted fixed points}
\label{sec:exp_vectorfield_regression}

This section empirically reports the ability of the proposed scheme to handle regression of a target velocity field which displays a collection of (stable or unstable) fixed points. In turn, this is aimed at providing a (non exhaustive) numerical validation of the universality results of Sections~\ref{theorem} and~\ref{sec:planting_node}. More concretely, we will consider supervised \emph{vector field regression} on a compact domain: given a target velocity field $F_\star:K\to\mathbb R^2$ and pointwise samples
$\{(x_i, F_\star(x_i))\}$, we shall train a one-hidden-layer Neural-ODE vector field
$F_\theta$ to approximate $F_\star$ while
exactly reproducing the fixed points of the target velocity field $F_\star$. The fixed-point constraints are imposed \emph{by construction} through the QR-based planting parametrization described in Sections~\ref{subsec:qr_planting} and~\ref{subsec:node_setup}, so that every iterate of the gradient-based optimization lies in the constrained hypothesis class.

\subsection{Constrained approximation model}
\label{subsec:exp_constrained_model}

We use the residual model \eqref{eq:node_model} with $n=2$, sigmoid activation and hidden width $m=256$, and with the output matrix $A$ given by the QR parametrization \eqref{eq:A1_residual_param}; the trainable parameters are $W$, $B$, $b$ and $d$.
Given a target field $F_\star$ on a compact set $K\subset\mathbb R^2$, we minimize the mean-squared
regression loss
\begin{equation}
\label{eq:mse_loss_exp}
\mathcal L(\theta)
=\frac{1}{N_{\mathrm{train}}}\sum_{i=1}^{N_{\mathrm{train}}}
\big\|F_\theta(x_i)-F_\star(x_i)\big\|_2^2,
\end{equation}
using samples $x_i$ drawn uniformly in $K$.
To check that the constraints are satisfied, we monitor the fixed-point residual
\begin{equation}
\label{eq:fp_residual_exp}
r_{\mathrm{fp}}(\theta)
:=\max_{l=1,\dots,C}\big\|F_\theta(\bar x^{(l)})\big\|_2,
\end{equation}
which should remain at numerical precision since the constraints are built into the parametrization. To assess spatial accuracy, we also evaluate the pointwise error on a dense grid,
\begin{equation}
\label{eq:pointwise_error_exp}
e(x):=\big\|F_\theta(x)-F_\star(x)\big\|_2,
\end{equation}
and visualize $e(x)$ as a log-scale heatmap. As a global accuracy measure we use the grid RMSE
\begin{equation}
\label{eq:grid_rmse_exp}
\mathrm{RMSE}_{\mathrm{grid}}(\theta):=\Big(\frac{1}{|G|}\sum_{x\in G}e(x)^2\Big)^{1/2},
\end{equation}
where $G$ is a uniform $250\times250$ grid on $K$. The grid is used only for evaluation. It enters neither the training loss nor the choice of the model (unless otherwise stated, all reported values refer to the last epoch, without early stopping), and it is independent of the training samples, which are drawn at random in $K$. The grid RMSE is therefore a test error, which covers uniformly the whole domain, including the regions close to the boundary of $K$.

\subsection{Experiment 1: generalized Lotka-Volterra vector field with resources competition}
\label{subsec:Lotka-Volterra_experiment}

Let $K=[-1,4]\times[-1,4]\subset\mathbb R^2$ and define the polynomial target vector field
$F_{\mathrm{base}}:\mathbb R^2\to\mathbb R^2$ by
\begin{equation}
\label{eq:F_base_exp}
F_{\mathrm{base}}(x,y)
=
\begin{pmatrix}
x(3-x)-2xy\\[2pt]
y(2-y)-xy
\end{pmatrix}.
\end{equation}
This is an instance of the Lotka-Volterra competition model, in which two species compete for a limited resource. With the chosen coefficients the two species cannot coexist stably, and which of them survives depends on the initial condition.

The specific field defined by Eq.~(\ref{eq:F_base_exp}) admits four equilibria inside $K$, namely $\bar x^{(1)}=(0,0)$, $\bar x^{(2)}=(0,2)$, $\bar x^{(3)}=(3,0)$ and $\bar x^{(4)}=(1,1)$. It is straightforward to verify by direct inspection that $F_{\mathrm{base}}(\bar x^{(l)})=0$, for all $l=1,\dots,4$.
We set $C=4$ and plant exactly these equilibria via \eqref{eq:fp_constraints_exp}.
We sample $N_{\mathrm{train}}=5\times10^4$ points uniformly in $K$ and minimize \eqref{eq:mse_loss_exp}
with Adam \cite{kingma2014adam}, using mini-batches of size $500$ for $400$ epochs, with a learning rate decreased from $10^{-3}$ to $0$ by cosine annealing. Since the learning rate vanishes at the end of training, the final values of the metrics are not dominated by epoch-to-epoch fluctuations.
At each optimization step we recompute $A$ using \eqref{eq:A1_residual_param}, ensuring the constraints
remain satisfied throughout training.

\begin{figure}[t]
    \centering
    \includegraphics[width=0.98\linewidth]{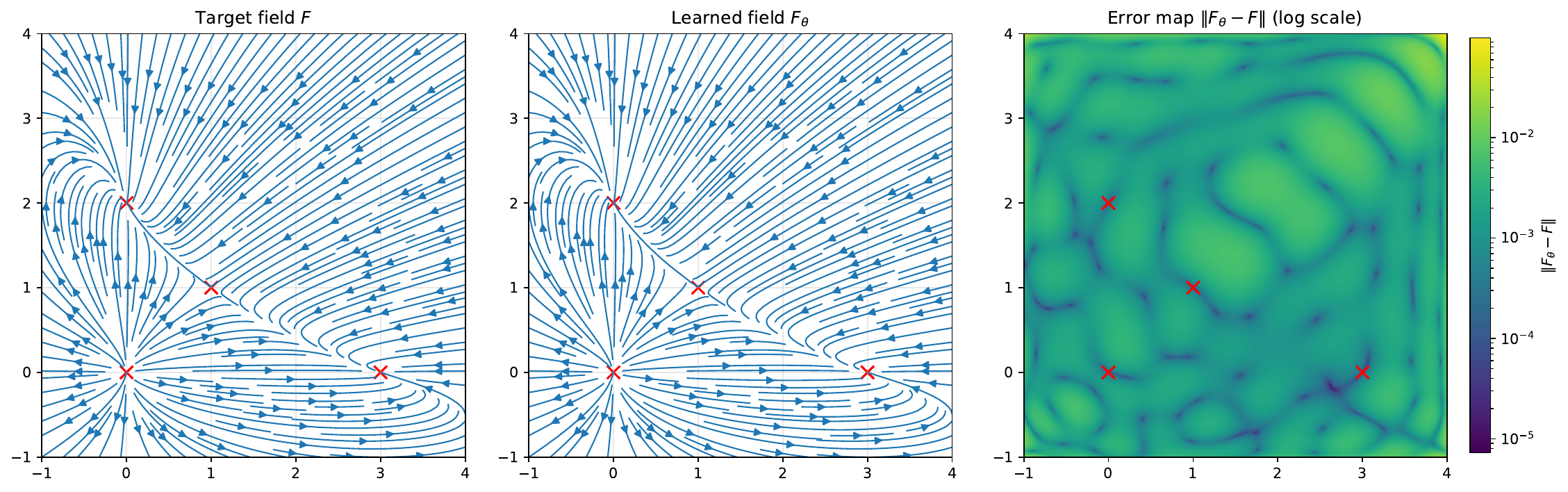}
    \caption{Vector field regression with four planted fixed points ($N_{\mathrm{train}}=5\times10^4$, seed $0$).
    (\emph{Left}) Target field $F_{\mathrm{base}}$ from \eqref{eq:F_base_exp}.
    (\emph{Center}) Learned constrained field $F_\theta$ from \eqref{eq:node_model} and \eqref{eq:A1_residual_param}.
    (\emph{Right}) Log-scale error heatmap $e(x)=\|F_\theta(x)-F_{\mathrm{base}}(x)\|_2$.
    Red crosses indicate the prescribed equilibria, which are matched (up to numerical precision) by construction.}
    \label{fig:exp_vectorfield_regression}
\end{figure}

To quantify the
variability due to random initialization, data sampling, and stochastic mini-batch optimization, we repeat the
experiment over $5$ independent random seeds. Since the errors are positive and, for some of the models considered below, vary by more than an order of magnitude across the seeds, in the rest of this experiment we summarize each error metric by its geometric mean $\mu_g$ and its geometric standard deviation $s_g$ over the seeds. Equivalently, $\log_{10}\mu_g$ and $\log_{10}s_g$ are the mean and the standard deviation of the base-$10$ logarithm of the metric. We write $\mu_g\,{\times}/{\div}\,s_g$, and show the interval $[\mu_g/s_g,\,\mu_g s_g]$ as a shaded band in the figures.
All runs were performed using PyTorch on a NVIDIA RTX A5500 with Data Loaders \texttt{num\_workers}=0. The wall-clock time per training run was
$1.92 \pm 0.04$ minutes (mean $\pm$ standard deviation over the $5$ seeds).
Figure~\ref{fig:exp_vectorfield_regression} shows the qualitative comparison between the target phase portrait
and the learned constrained phase portrait, together with the spatial distribution of the error $e(x)$.
The constrained model accurately recovers the global structure of $F_{\mathrm{base}}$ on $K$ while preserving
the prescribed equilibria. Quantitatively, the grid RMSE \eqref{eq:grid_rmse_exp} is
$3.75{\times}10^{-3}\,{\times}/{\div}\,1.06$ and the maximum pointwise error $\max_{x\in G}e(x)$ is
$1.03{\times}10^{-1}\,{\times}/{\div}\,1.06$. The fixed-point residual remains small across runs,
$4.02{\times}10^{-6}\,{\times}/{\div}\,1.84$, confirming that the QR-planted parametrization
prevents drift away from the prescribed hard constraints while retaining enough flexibility to fit the target
vector field.

\paragraph{Comparison with unconstrained and penalized MLPs.}
To isolate the effect of the hard constraints, we compare the QR-constrained model with two baselines that share its architecture \eqref{eq:node_model}, initialization scheme and training data. In the first one, which we call \emph{free MLP}, a shallow network is trained without enforcing the fixed points. The second one is the same network trained with a soft penalty on the prescribed equilibria, namely by minimizing
\begin{equation}
\label{eq:penalty_loss_exp}
\mathcal L_\lambda(\theta)=\frac12\,\mathcal L(\theta)+\frac{\lambda}{C}\sum_{l=1}^{C}\big\|F_\theta(\bar x^{(l)})\big\|_2^2,
\qquad \lambda\in\{1,10,50\},
\end{equation}
where the factor $1/2$ arises because the implementation averages the squared error over the two components of the field. All the models are trained with the protocol described above ($N_{\mathrm{train}}=5\times10^4$, mini-batches of size $500$, $400$ epochs, cosine-annealed learning rate, the same $5$ random seeds); for the QR-constrained model, these are the runs reported above. After each epoch we record the grid RMSE \eqref{eq:grid_rmse_exp} and the fixed-point residual \eqref{eq:fp_residual_exp}, which are shown in Figures~\ref{fig:lv_grid_rmse} and~\ref{fig:lv_fp_residual}, respectively.

\begin{figure}[t]
    \centering
    \includegraphics[width=0.8\linewidth]{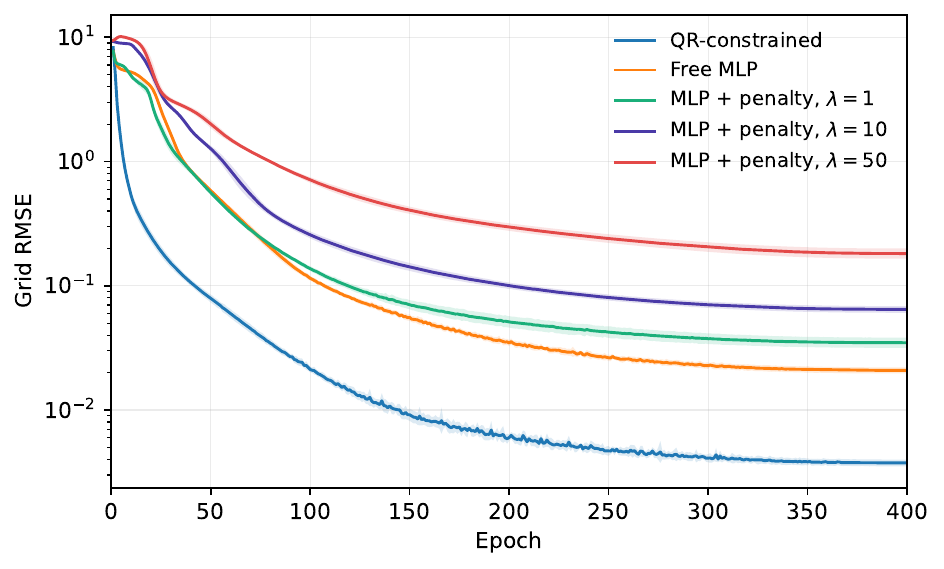}
    \caption{Lotka-Volterra experiment: grid RMSE \eqref{eq:grid_rmse_exp} during training for the QR-constrained model, a free MLP and MLPs trained on the penalized loss \eqref{eq:penalty_loss_exp}, $\lambda\in\{1,10,50\}$ ($N_{\mathrm{train}}=5\times10^4$, cosine-annealed learning rate, $5$ seeds). Lines: geometric mean $\mu_g$ over the seeds; shaded bands: $\mu_g\,{\times}/{\div}\,s_g$, i.e.\ $\pm1$ standard deviation of the base-$10$ logarithm of the metric across the seeds.}
    \label{fig:lv_grid_rmse}
\end{figure}

\begin{figure}[t]
    \centering
    \includegraphics[width=0.8\linewidth]{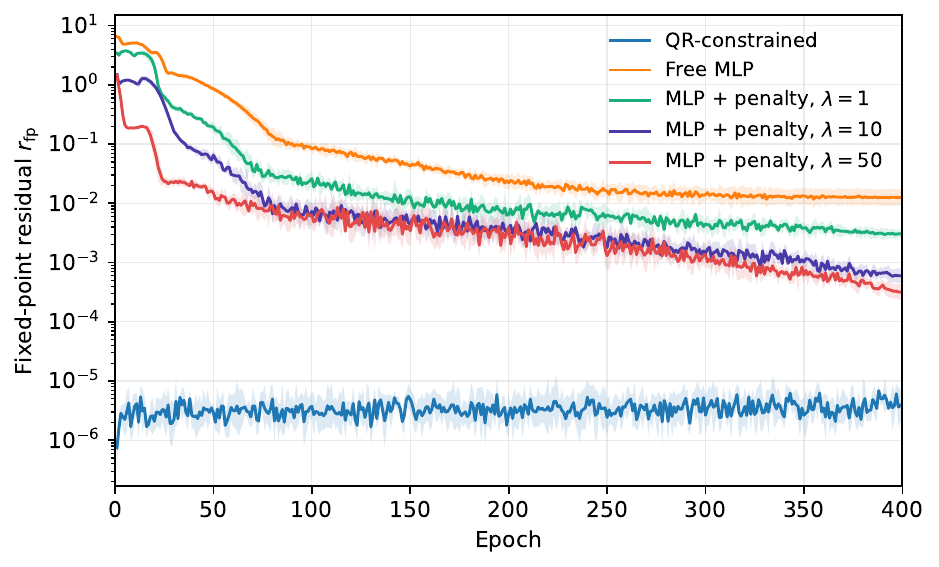}
    \caption{Lotka-Volterra experiment: fixed-point residual $r_{\mathrm{fp}}$ \eqref{eq:fp_residual_exp} during training, for the same runs as in Figure~\ref{fig:lv_grid_rmse}. Lines and bands as in Figure~\ref{fig:lv_grid_rmse}.}
    \label{fig:lv_fp_residual}
\end{figure}

The QR-constrained model is the most accurate throughout training (Figure~\ref{fig:lv_grid_rmse}): from the second epoch on, its grid RMSE is lower than that of every baseline, and at the last epoch it is more than five times lower than that of the free MLP.
The fixed-point residual of the QR-constrained model stays at the level of single-precision round-off at every epoch and for every seed (Figure~\ref{fig:lv_fp_residual}), as expected from constraints that hold for every value of the parameters. In the free MLP the equilibria are reproduced only through the regression of $F_{\mathrm{base}}$, so the residual decreases together with the regression error and at the last epoch is of the same order as the grid RMSE. The penalty lowers the residual, which nevertheless remains about two to three orders of magnitude above that of the QR-constrained model.

The penalty therefore trades accuracy for constraint satisfaction: as $\lambda$ increases, the residual decreases but the grid RMSE increases, and at the last epoch the latter ranges from $1.7$ ($\lambda=1$) to $8.7$ ($\lambda=50$) times that of the free MLP. This happens although the target field vanishes at the prescribed points, so that the two terms of \eqref{eq:penalty_loss_exp} share a common minimizer. The QR parametrization avoids this trade-off. Every field of the constrained family is also a field of the free architecture, so the free MLP is at least as expressive as the QR-constrained model. The observed accuracy gap therefore cannot be attributed to greater expressive power, but rather reflects the effects of the constrained parametrization on optimization and implicit regularization.

The same qualitative picture is observed with a constant learning rate $10^{-3}$, for which the final grid RMSE is $3.0\times10^{-3}$ for the QR-constrained model, $1.6\times10^{-2}$ for the free MLP, and between $2.3\times10^{-2}$ and $1.4\times10^{-1}$ for the penalized MLPs.

\paragraph{Dependence on the dataset size.}
We finally study how the comparison depends on the amount of available data, and whether the prior knowledge is equally effective when it is supplied in a different way. To this end, we add to the baselines an \emph{augmented MLP}, i.e.\ the free MLP trained on the available samples together with copies of the $C$ equilibria $\bar x^{(l)}$ with target $F_\star(\bar x^{(l)})=0$, replicated so that they make up a fraction $\rho$ of the training set. Here the prior knowledge enters as additional data, not as a constraint. We draw $N_{\mathrm{data}}\in\{200,\,500,\,10^3,\,10^4,\,5\times10^4\}$ samples uniformly in $K$, and use $80\%$ of them for training and the remaining $20\%$ for validation. In order to compare the models rather than the training budgets, every run is trained until the validation error no longer improves. We use Adam with initial learning rate $10^{-3}$ and mini-batches of size $500$, or full batches when the training set is smaller. Every $10^3$ optimizer steps the validation MSE is computed, and the learning rate is halved when it has not improved for $10$ consecutive evaluations; training stops when the learning rate falls below $10^{-6}$, i.e.\ after $10$ halvings, and the checkpoint with the lowest validation MSE is retained. All the runs reported below stop by this criterion, after between $1.25\times10^5$ and $4.4\times10^6$ optimizer steps (median $7.4\times10^5$). The penalty weight $\lambda\in\{10^{-2},10^{-1},1,10\}$ in \eqref{eq:penalty_loss_exp} and the fraction $\rho\in\{0.1,0.25,0.5,0.75\}$ are selected for each $N_{\mathrm{data}}$ by the geometric mean of the lowest validation MSE over three seeds (the limits $\lambda\to0$ and $\rho\to0$ correspond to the free MLP). The selected values are $\lambda=10^{-1}$ for $N_{\mathrm{data}}\in\{200,\,5\times10^4\}$ and $\lambda=10^{-2}$ otherwise, and $\rho=0.5$ for $N_{\mathrm{data}}=200$, $\rho=0.25$ for $N_{\mathrm{data}}=5\times10^4$ and $\rho=0.1$ otherwise. With these values, the experiment is repeated over $10$ random seeds. All the models therefore share the architecture, the number of trainable parameters, the data and the training protocol, and differ only in the way in which the prior knowledge on the equilibria is supplied.

\begin{figure}[t]
    \centering
    \includegraphics[width=0.8\linewidth]{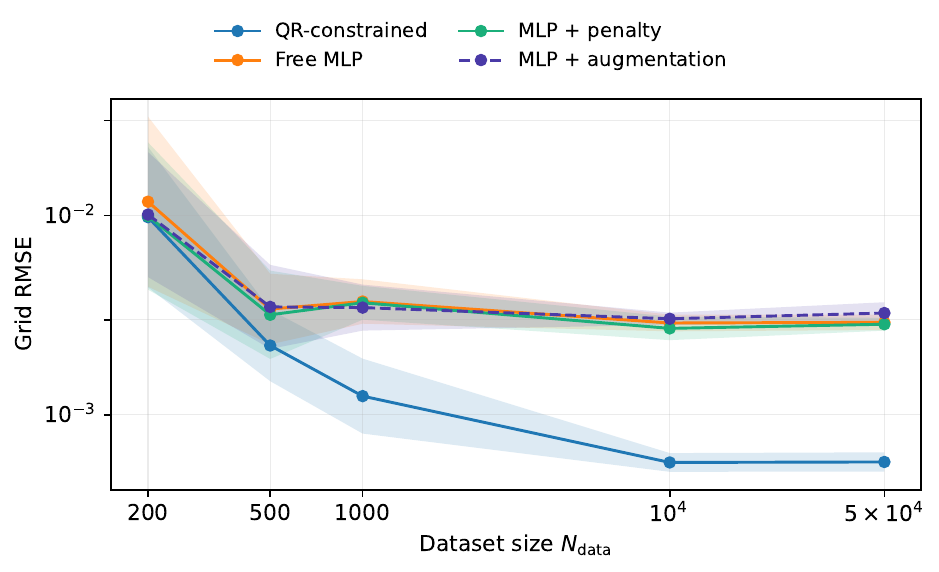}
    \caption{Lotka-Volterra experiment: grid RMSE \eqref{eq:grid_rmse_exp} of the retained checkpoint as a function of the dataset size $N_{\mathrm{data}}$ ($80\%$ training, $20\%$ validation), for models trained until the validation error no longer improves, $10$ seeds. The curves of the three baselines nearly overlap. Lines and bands as in Figure~\ref{fig:lv_grid_rmse}.}
    \label{fig:lv_dataset_size}
\end{figure}

\begin{figure}[t]
    \centering
    \includegraphics[width=\linewidth]{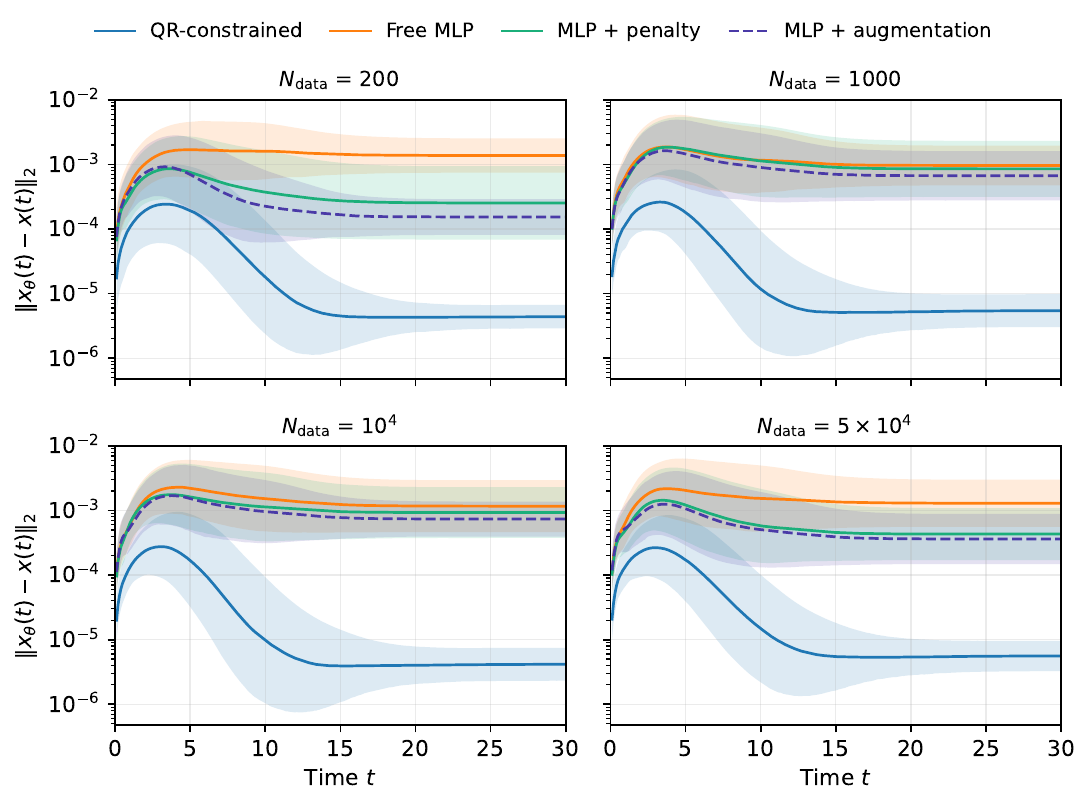}
    \caption{Lotka-Volterra experiment: error $\|x_\theta(t)-x(t)\|_2$ between the trajectories of the learned fields and those of $F_{\mathrm{base}}$, integrated with the same fourth-order Runge-Kutta scheme ($\Delta t=0.01$) from $500$ initial conditions drawn uniformly in $[0.2,3.5]^2$, for the models of Figure~\ref{fig:lv_dataset_size} with $N_{\mathrm{data}}\in\{200,\,10^3,\,10^4,\,5\times10^4\}$. Lines and bands as in Figure~\ref{fig:lv_grid_rmse}, with the statistics taken over the initial conditions and the seeds.}
    \label{fig:lv_trajectory_error}
\end{figure}

\begin{table}[t]
    \centering
    \footnotesize
    \setlength{\tabcolsep}{4pt}
    \caption{Lotka-Volterra experiment, dependence on the dataset size: geometric mean ${\times}/{\div}$ geometric standard deviation over $10$ seeds at the retained checkpoint. Local RMSE: RMSE on the points of $G$ within distance $0.25$ of a prescribed equilibrium. Jacobian error: ${\max_{l}\|DF_\theta(\bar x^{(l)})-DF_\star(\bar x^{(l)})\|_2}$, where $DF$ denotes the Jacobian matrix and $\|\cdot\|_2$ the spectral norm. Trajectory error: $\|x_\theta(30)-x(30)\|_2$ for the trajectories of Figure~\ref{fig:lv_trajectory_error}, geometric mean over the $500$ initial conditions. In bold, the lowest geometric mean of each column for each $N_{\mathrm{data}}$.}
    \label{tab:lv_dataset_size}
    \resizebox{\linewidth}{!}{%
    \begin{tabular}{lccccc}
        \toprule
        Method & Grid RMSE & Local RMSE & $r_{\mathrm{fp}}$ & Jacobian error & Traj.\ error ($t=30$) \\
        \midrule
        \multicolumn{6}{l}{$N_{\mathrm{data}}=200$}\\
        QR-constrained & {\boldmath$9.8\times10^{-3}\,{\times}/{\div}\,2.25$} & {\boldmath$3.0\times10^{-4}\,{\times}/{\div}\,2.38$} & {\boldmath$4.0\times10^{-6}\,{\times}/{\div}\,2.05$} & {\boldmath$3.5\times10^{-3}\,{\times}/{\div}\,2.58$} & {\boldmath$4.4\times10^{-6}\,{\times}/{\div}\,1.29$} \\
        Free MLP & $1.2\times10^{-2}\,{\times}/{\div}\,2.67$ & $1.5\times10^{-3}\,{\times}/{\div}\,1.82$ & $2.2\times10^{-3}\,{\times}/{\div}\,1.85$ & $6.9\times10^{-3}\,{\times}/{\div}\,1.94$ & $1.4\times10^{-3}\,{\times}/{\div}\,1.39$ \\
        MLP + penalty & $9.9\times10^{-3}\,{\times}/{\div}\,2.34$ & $6.9\times10^{-4}\,{\times}/{\div}\,1.94$ & $7.1\times10^{-4}\,{\times}/{\div}\,2.21$ & $5.3\times10^{-3}\,{\times}/{\div}\,1.43$ & $2.5\times10^{-4}\,{\times}/{\div}\,3.31$ \\
        MLP + augmentation & $1.0\times10^{-2}\,{\times}/{\div}\,2.06$ & $6.2\times10^{-4}\,{\times}/{\div}\,1.25$ & $2.8\times10^{-4}\,{\times}/{\div}\,1.58$ & $6.3\times10^{-3}\,{\times}/{\div}\,1.41$ & $1.5\times10^{-4}\,{\times}/{\div}\,1.66$ \\
        \midrule
        \multicolumn{6}{l}{$N_{\mathrm{data}}=1000$}\\
        QR-constrained & {\boldmath$1.2\times10^{-3}\,{\times}/{\div}\,1.54$} & {\boldmath$2.4\times10^{-4}\,{\times}/{\div}\,1.24$} & {\boldmath$4.8\times10^{-6}\,{\times}/{\div}\,1.91$} & {\boldmath$2.9\times10^{-3}\,{\times}/{\div}\,1.36$} & {\boldmath$5.4\times10^{-6}\,{\times}/{\div}\,1.57$} \\
        Free MLP & $3.7\times10^{-3}\,{\times}/{\div}\,1.29$ & $1.3\times10^{-3}\,{\times}/{\div}\,1.28$ & $1.7\times10^{-3}\,{\times}/{\div}\,1.35$ & $6.8\times10^{-3}\,{\times}/{\div}\,1.31$ & $9.6\times10^{-4}\,{\times}/{\div}\,1.63$ \\
        MLP + penalty & $3.6\times10^{-3}\,{\times}/{\div}\,1.22$ & $1.1\times10^{-3}\,{\times}/{\div}\,1.38$ & $1.4\times10^{-3}\,{\times}/{\div}\,1.57$ & $6.2\times10^{-3}\,{\times}/{\div}\,1.19$ & $8.5\times10^{-4}\,{\times}/{\div}\,1.83$ \\
        MLP + augmentation & $3.5\times10^{-3}\,{\times}/{\div}\,1.30$ & $9.1\times10^{-4}\,{\times}/{\div}\,1.27$ & $1.1\times10^{-3}\,{\times}/{\div}\,1.39$ & $6.2\times10^{-3}\,{\times}/{\div}\,1.39$ & $6.7\times10^{-4}\,{\times}/{\div}\,1.83$ \\
        \midrule
        \multicolumn{6}{l}{$N_{\mathrm{data}}=10^4$}\\
        QR-constrained & {\boldmath$5.8\times10^{-4}\,{\times}/{\div}\,1.11$} & {\boldmath$2.4\times10^{-4}\,{\times}/{\div}\,1.22$} & {\boldmath$3.2\times10^{-6}\,{\times}/{\div}\,1.50$} & {\boldmath$2.7\times10^{-3}\,{\times}/{\div}\,1.34$} & {\boldmath$4.2\times10^{-6}\,{\times}/{\div}\,1.53$} \\
        Free MLP & $2.9\times10^{-3}\,{\times}/{\div}\,1.10$ & $1.3\times10^{-3}\,{\times}/{\div}\,1.21$ & $1.7\times10^{-3}\,{\times}/{\div}\,1.30$ & $6.1\times10^{-3}\,{\times}/{\div}\,1.14$ & $1.2\times10^{-3}\,{\times}/{\div}\,1.68$ \\
        MLP + penalty & $2.7\times10^{-3}\,{\times}/{\div}\,1.15$ & $1.0\times10^{-3}\,{\times}/{\div}\,1.30$ & $1.2\times10^{-3}\,{\times}/{\div}\,1.35$ & $6.1\times10^{-3}\,{\times}/{\div}\,1.28$ & $9.3\times10^{-4}\,{\times}/{\div}\,1.73$ \\
        MLP + augmentation & $3.0\times10^{-3}\,{\times}/{\div}\,1.07$ & $9.6\times10^{-4}\,{\times}/{\div}\,1.19$ & $1.1\times10^{-3}\,{\times}/{\div}\,1.31$ & $6.8\times10^{-3}\,{\times}/{\div}\,1.21$ & $7.4\times10^{-4}\,{\times}/{\div}\,1.42$ \\
        \midrule
        \multicolumn{6}{l}{$N_{\mathrm{data}}=5\times10^4$}\\
        QR-constrained & {\boldmath$5.8\times10^{-4}\,{\times}/{\div}\,1.12$} & {\boldmath$2.4\times10^{-4}\,{\times}/{\div}\,1.14$} & {\boldmath$4.1\times10^{-6}\,{\times}/{\div}\,1.61$} & {\boldmath$2.8\times10^{-3}\,{\times}/{\div}\,1.31$} & {\boldmath$5.6\times10^{-6}\,{\times}/{\div}\,1.62$} \\
        Free MLP & $2.9\times10^{-3}\,{\times}/{\div}\,1.09$ & $1.3\times10^{-3}\,{\times}/{\div}\,1.30$ & $1.7\times10^{-3}\,{\times}/{\div}\,1.41$ & $6.3\times10^{-3}\,{\times}/{\div}\,1.09$ & $1.3\times10^{-3}\,{\times}/{\div}\,1.55$ \\
        MLP + penalty & $2.8\times10^{-3}\,{\times}/{\div}\,1.08$ & $7.9\times10^{-4}\,{\times}/{\div}\,1.16$ & $7.4\times10^{-4}\,{\times}/{\div}\,1.28$ & $6.5\times10^{-3}\,{\times}/{\div}\,1.17$ & $4.3\times10^{-4}\,{\times}/{\div}\,1.95$ \\
        MLP + augmentation & $3.2\times10^{-3}\,{\times}/{\div}\,1.13$ & $8.6\times10^{-4}\,{\times}/{\div}\,1.18$ & $7.6\times10^{-4}\,{\times}/{\div}\,1.32$ & $7.7\times10^{-3}\,{\times}/{\div}\,1.25$ & $3.6\times10^{-4}\,{\times}/{\div}\,1.63$ \\
        \bottomrule
    \end{tabular}}
\end{table}

When every model is trained until it no longer improves, the grid RMSE of the QR-constrained model decreases up to $N_{\mathrm{data}}=10^4$, while that of the three baselines saturates from $N_{\mathrm{data}}=500$ on (Figure~\ref{fig:lv_dataset_size}). At $N_{\mathrm{data}}=200$ the four models are equivalent in this metric, while for $N_{\mathrm{data}}\ge10^3$ the QR-constrained model has the lowest grid RMSE for each of the $10$ seeds. The penalty and the augmentation change the grid RMSE of the free MLP by less than $20\%$ at every $N_{\mathrm{data}}$, so that, in terms of global accuracy, supplying the prior knowledge as a penalty or as additional data brings little benefit.

To assess the learned dynamics, we integrate the learned fields and the true field $F_{\mathrm{base}}$ with the same scheme and from the same $500$ initial conditions, up to $t=30$. The true flow is bistable: $\bar x^{(2)}=(0,2)$ and $\bar x^{(3)}=(3,0)$ are stable nodes, $\bar x^{(4)}=(1,1)$ is a saddle, and $\bar x^{(1)}=(0,0)$ is an unstable node; $314$ initial conditions reach $\bar x^{(2)}$ and $186$ reach $\bar x^{(3)}$. All the models reproduce this qualitative behavior: for every model, seed and dataset size, each of the $500$ learned trajectories ends, at $t=30$, closest to the same equilibrium as the true one. The quantitative accuracy of the dynamics, however, differs markedly. The trajectory error of the QR-constrained model has a lower transient peak than that of the baselines, and then decays as the trajectories reach the planted equilibria (Figure~\ref{fig:lv_trajectory_error}). The error of the baselines, by contrast, does not vanish, because their trajectories converge to equilibria that are displaced from the true ones. This holds at every dataset size and for each seed, and the trajectory error of the QR-constrained model does not depend on $N_{\mathrm{data}}$.

Table~\ref{tab:lv_dataset_size} summarizes these results. As for the dynamics, the accuracy of the QR-constrained model near the equilibria does not depend on the dataset size: its local RMSE (geometric mean over the seeds) is between $1.6\times10^{-4}$ and $3.1\times10^{-4}$ at every $N_{\mathrm{data}}$. This is about $5$ times lower than the local RMSE of the free MLP, and the fixed-point residual remains at round-off level, against $10^{-4}$--$10^{-3}$ for the baselines. The Jacobian of the learned field at the prescribed equilibria, which the constraints leave free, is also more accurate for the QR-constrained model, although by a smaller factor. In all the runs the eigenvalues of the learned Jacobians are real and have the same signs as the true ones. The type of every equilibrium is therefore reproduced. Their relative error is $(2$--$3)\times10^{-3}$ for the QR-constrained model and between $4.7\times10^{-3}$ and $1.0\times10^{-2}$ for the baselines. Near the equilibria and along the long-term dynamics, the QR-constrained model trained on $N_{\mathrm{data}}=200$ samples is therefore more accurate than every baseline trained on $250$ times more data; away from the equilibria, it is the only model whose error benefits from additional data.

Another experiment of regression of a dynamical system involving a limit cycle behaviour is reported in Appendix~\ref{subsec:exp_vectorfield_limitcycle}. In Appendix~\ref{app:traj_training} we repeat the comparison with the models trained on observed trajectories, by backpropagation through the ODE solver; also in this setting, the QR-constrained model has the lowest trajectory error.

\section{Conclusions}
\label{Conclusions}

In this work we have shown how prior knowledge on the equilibria of an unknown vector field can be injected directly into the parameters of the Neural-ODE that approximates it. For fixed hidden parameters, the conditions that the field vanishes, or more generally takes prescribed values, at finitely many points are linear in the output matrix. The pseudoinverse parametrization, implemented through a QR factorization, splits the output matrix into a component fixed by the constraints and a free component that shapes the field elsewhere. The constraints therefore hold exactly for every value of the trainable parameters, and hence at every step of a training that updates all of them. Neither the loss nor the output of the network is modified. Since the parametrization represents every admissible output matrix, Theorem~\ref{thm:main} implies that the resulting family, including the residual Neural-ODE used in the experiments, is dense in the continuous vector fields compatible with the constraints, for any continuous nonpolynomial activation.

In the experiments on a Lotka-Volterra competition model, the exact constraints also improve the accuracy of the learned field. We compared it with networks of the same architecture trained on the same data, either alone, with a soft penalty on the equilibria, or with the equilibria added to the training set. Except on the smallest dataset, where all models perform alike, the constrained model has the lowest error over the whole domain, both at a fixed training budget and when trained until convergence. The constrained model continues to benefit from additional data over a wider range of dataset sizes than the baselines. The gap is larger close to the equilibria and in the long-time behavior of the trajectories: the trajectories of the baselines converge to displaced equilibria, and there the constrained model is more accurate than every baseline, even when trained on far less data. Supplied as a penalty or as additional data, the same prior knowledge improves the model only near the equilibria and, with large penalty weights, at the expense of global accuracy. As discussed in Section~\ref{subsec:Lotka-Volterra_experiment}, the constrained family is contained in the free one, so that these differences originate in training. The limit-cycle experiment of Appendix~\ref{subsec:exp_vectorfield_limitcycle} shows, in addition, that the constraint does not prevent the model from learning non-equilibrium structures such as a stable periodic orbit.

\bibliographystyle{unsrtnat}
\bibliography{references}
\newpage

\appendix

\section{An invertible evaluation matrix}
\label{app:lemma}

This appendix proves the lemma used in the proof of Theorem~\ref{thm:main} to select the correction neurons.

\begin{lemma}[Invertible evaluation matrix]
\label{lem:invertible_M}
Let $n\ge1$, let $C\ge1$, and let $\bar x^{(1)},\dots,\bar x^{(C)}\in\mathbb R^n$ be distinct. Let $\sigma:\mathbb R\to\mathbb R$ be continuous and not coincide with a polynomial on all of $\mathbb R$. There exist $C$ parameter pairs $(w_k,\beta_k)\in\mathbb R^n\times\mathbb R$ such that
\begin{equation}
\label{eq:M_def}
M_{lk}=\sigma(w_k^\top\bar x^{(l)}+\beta_k),\qquad l,k=1,\dots,C,
\end{equation}
defines an invertible matrix $M\in\mathbb R^{C\times C}$. These pairs can be chosen independently of any target values or approximation tolerance.
\end{lemma}

\begin{proof}
Let
\[
X=\{\bar x^{(1)},\dots,\bar x^{(C)}\}\subset\mathbb R^n
\]
be the finite (hence compact) set of prescribed input points. For each single neuron, define its vector of values at these points:
\[
v(w,\beta)=
\begin{bmatrix}
\sigma(w^\top\bar x^{(1)}+\beta)\\
\vdots\\
\sigma(w^\top\bar x^{(C)}+\beta)
\end{bmatrix}\in\mathbb R^C,
\qquad
V=\operatorname{span}\{v(w,\beta):w\in\mathbb R^n,\ \beta\in\mathbb R\}.
\]

To establish density, fix an arbitrary vector $y=(y_1,\dots,y_C)^\top\in\mathbb R^C$. Associate with it the scalar function
\[
f_y:X\to\mathbb R,\qquad f_y(\bar x^{(l)})=y_l\quad(l=1,\dots,C).
\]
Thus $f_y$ is simply the continuous function whose table assigns the value $y_l$ to the input point $\bar x^{(l)}$. The distinctness of the points allows the entries of $y$ to be chosen independently, and
\[
\big(f_y(\bar x^{(1)}),\dots,f_y(\bar x^{(C)})\big)^\top=y.
\]
The universal approximation theorem \cite{leshno1993multilayer}, applied on $X$, therefore gives, for every $\eta>0$, a finite network
\[
g(x)=\sum_{s=1}^p a_s\,\sigma(u_s^\top x+\gamma_s)
\]
such that
\[
\max_{1\le l\le C}\big|g(\bar x^{(l)})-f_y(\bar x^{(l)})\big|=\max_{1\le l\le C}\big|g(\bar x^{(l)})-y_l\big|<\eta.
\]
Notice now that the vector filled with the evaluations of $g(x)$ over the prescribed points belongs to $V$ by definition of linear combination:
\[
\big(g(\bar x^{(1)}),\dots,g(\bar x^{(C)})\big)^\top=\sum_{s=1}^p a_s\,v(u_s,\gamma_s) \in V.
\]
In particular, choosing $\eta=1/j$ and denoting by $g_j$ the corresponding network, the vectors $v_j=\big(g_j(\bar x^{(1)}),\dots,g_j(\bar x^{(C)})\big)^\top\in V$ satisfy $\|v_j-y\|_2<\sqrt C/j$, hence $v_j\to y$ as $j\to \infty$. This approximation controls all $C$ components simultaneously. Since $y$ was arbitrary, $V$ is dense in $\mathbb R^C$, meaning
\[
\overline V=\mathbb R^C.
\]
Here $\overline V$ denotes the closure of $V$, namely the set of limits in $\mathbb R^C$ of convergent sequences with terms in $V$. Density alone places every $y$ in $\overline V$; to conclude that $y$ itself belongs to $V$, we also establish closedness:

Because $V$ is a linear subspace of $\mathbb R^C$, its dimension is at most $C$. Choose an orthonormal basis $q_1,\dots,q_r$ of $V$, where $r\le C$. Take any sequence $z_j\in V$ such that $z_j\to z\in\mathbb R^C$. Convergence is the hypothesis here: the limit $z$ is a priori only an element of the ambient space $\mathbb R^C$, and we show that it lies in $V$. Then
\[
z_j=\sum_{i=1}^r\langle z_j,q_i\rangle q_i
\quad\Longrightarrow\quad
z=\sum_{i=1}^r\langle z,q_i\rangle q_i\in V.
\]
The equality follows by continuity of the inner product and of finite sums. Thus every convergent sequence in $V$ has its limit in $V$, which proves $\overline V=V$. Combining this closedness with the density proved above gives
\[
V=\overline V=\mathbb R^C.
\]

A spanning family of $\mathbb R^C$ contains a basis of $C$ vectors. Hence we can select $C$ individual vectors $v(w_k,\beta_k)$ that are linearly independent. Taking them as the columns of $M$ gives \eqref{eq:M_def} and proves its invertibility. Only the prescribed points and the activation were used in this selection.
\end{proof}

\section{Experiment 2: limit-cycle system with a planted equilibrium}
\label{subsec:exp_vectorfield_limitcycle}

We now test the same constrained parametrization on a qualitatively different phase portrait: a system with a single prescribed (unstable) fixed point and a stable periodic orbit. This experiment is meant to probe whether hard equilibrium constraints interfere with learning
attractors that are not equilibria (here, a limit cycle). In particular, the planted constraint fixes the behavior
at one point in state space, but leaves the remaining degrees of freedom available to approximate the global
geometry of the flow, including periodic behavior.


Let $K=[-1,4]\times[-1,4]\subset\mathbb R^2$ and consider the (Brusselator-like) target vector field
$F_{\mathrm{lc}}:\mathbb R^2\to\mathbb R^2$ defined by
\begin{equation}
\label{eq:F_lc_base_exp}
F_{\mathrm{lc}}(x,y)
=
\begin{pmatrix}
-x + a\,y + x^2y\\[2pt]
b - a\,y - x^2y
\end{pmatrix},
\qquad (a>0,\ b>0).
\end{equation}
In the experiment we fix
\[
a=0.06,\qquad b=0.6.
\]

A direct computation shows that $F_{\mathrm{lc}}$ has a unique equilibrium point
\begin{equation}
\label{eq:fixed_points_lc_exp}
\bar x^{(1)} = \Big(b,\ \frac{b}{a+b^2}\Big),
\end{equation}
and $F_{\mathrm{lc}}(\bar x^{(1)})=0$. For these parameter values, the above fixed point is an unstable spiral, and trajectories from a broad set of initial conditions in $K$ converge, spiraling, toward a stable periodic orbit, yielding a limit-cycle phase portrait.


We approximate $F_{\mathrm{lc}}$ on $K$ with the same constrained Neural-ODE model \eqref{eq:node_model}
(with $m=256$). Here we set $C=1$ and plant the equilibrium \eqref{eq:fixed_points_lc_exp} by enforcing
\begin{equation}
\label{eq:fp_constraints_lc_exp}
F_\theta(\bar x^{(1)})=0,
\end{equation}
using the QR-based construction \eqref{eq:A1_residual_param} (the case $C=1$ reduces to a normalized one-dimensional
projector, but the implementation is identical).

We sample $N_{\mathrm{train}}=10^6$ points uniformly in $K$ and minimize the regression loss
\begin{equation}
\label{eq:mse_loss_lc_exp}
\mathcal L(\theta)
=\frac{1}{N_{\mathrm{train}}}\sum_{i=1}^{N_{\mathrm{train}}}
\big\|F_\theta(x_i)-F_{\mathrm{lc}}(x_i)\big\|_2^2,
\end{equation}
training with Adam (learning rate $10^{-3}$) and mini-batches of size $500$ for $200$ epochs.
We monitor the planted-equilibrium residual
\begin{equation}
\label{eq:fp_residual_lc_exp}
r_{\mathrm{fp}}(\theta)
:=\big\|F_\theta(\bar x^{(1)})\big\|_2,
\end{equation}
and visualize the log-scale pointwise error
\begin{equation}
\label{eq:pointwise_error_lc_exp}
e(x):=\big\|F_\theta(x)-F_{\mathrm{lc}}(x)\big\|_2.
\end{equation}


Figure~\ref{fig:exp_vectorfield_limitcycle} compares the target flow, the learned constrained approximation,
and the resulting error heatmap. The learned field reproduces the qualitative geometry of the limit-cycle phase
portrait on $K$, while the equilibrium constraint at $\bar x^{(1)}$ is satisfied exactly at every epoch
(red cross). In this example, the hard fixed-point constraint does not prevent the approximation of a vector field whose attractor is a periodic orbit; the constraint fixes the field at the planted equilibrium and leaves enough freedom to fit the dynamics around it.

\begin{figure}[t]
    \centering
    \includegraphics[width=0.98\linewidth]{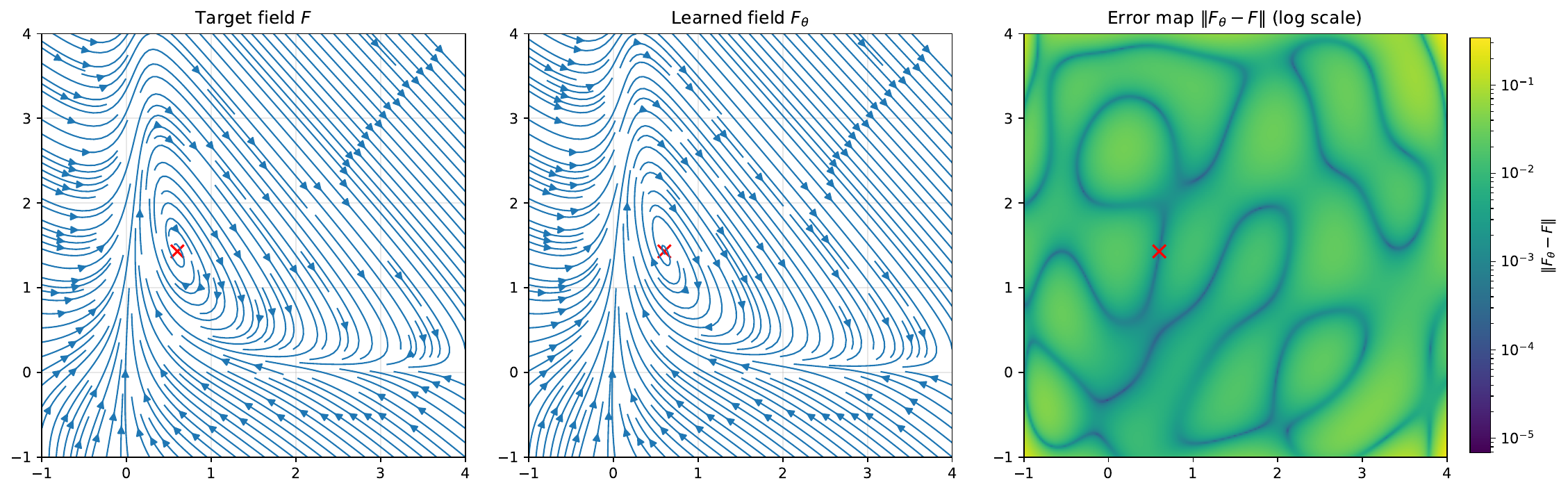}
    \caption{Vector field regression for a limit-cycle system with a planted equilibrium.
    (\emph{Left}) Target field $F_{\mathrm{lc}}$ from \eqref{eq:F_lc_base_exp}.
    (\emph{Center}) Learned constrained field $F_\theta$.
    (\emph{Right}) Log-scale error heatmap $e(x)=\|F_\theta(x)-F_{\mathrm{lc}}(x)\|_2$.
    The red cross marks the prescribed equilibrium \eqref{eq:fixed_points_lc_exp}, which is matched (up to numerical precision) by construction.}
    \label{fig:exp_vectorfield_limitcycle}
\end{figure}

\section{Training on trajectories through the ODE solver}
\label{app:traj_training}

In Section~\ref{sec:exp_vectorfield_regression} the models are trained by regression on samples of the target vector field. When a Neural-ODE is used as a proxy for an unknown dynamical system, however, the vector field is usually not observed directly: the data are trajectories, and the network is trained by integrating the learned field from the observed states and comparing the result with the subsequent observations, with the gradient propagated through the ODE solver~\cite{chen2018neuralode}. In this appendix we repeat the comparison of Section~\ref{subsec:Lotka-Volterra_experiment} in this setting, with the same target field \eqref{eq:F_base_exp}, the same four planted equilibria and the same architecture \eqref{eq:node_model}, for the QR-constrained model, the free MLP and the MLP trained with the penalty of \eqref{eq:penalty_loss_exp}.

\paragraph{Data.} For each seed we draw $100$ initial conditions uniformly in $[0,4]^2$ and integrate $F_{\mathrm{base}}$ up to $T=10$, recording the state every $\Delta t_{\mathrm{obs}}=0.1$, so that each trajectory consists of $101$ observed states. The reference solutions are computed in double precision with the fourth-order Runge-Kutta scheme (RK4) and step $10^{-4}$; halving the step changes them by less than $10^{-12}$. The square $[0,4]^2\subset K$ is forward invariant for $F_{\mathrm{base}}$, so that all the trajectories remain in $K$. We use $80$ trajectories for training and the remaining $20$ for validation.

\paragraph{Training loss.} Let $x_{n,k}$ be the observed state of trajectory $n$ at time $k\Delta t_{\mathrm{obs}}$, and let $\Phi_\theta^{j}(x)$ denote the state reached by integrating $F_\theta$ from $x$ over the time $j\Delta t_{\mathrm{obs}}$. We split the trajectories into windows of $h=10$ consecutive observations, i.e.\ the pairs $(n,k)$ with $0\le k\le 100-h$, and minimize
\begin{equation}
\label{eq:traj_loss}
\mathcal L_{\mathrm{traj}}(\theta)=\frac{1}{|\mathcal W|\,h}\sum_{(n,k)\in\mathcal W}\,\sum_{j=1}^{h}\big\|\Phi_\theta^{j}(x_{n,k})-x_{n,k+j}\big\|_2^2,
\end{equation}
where $\mathcal W$ is the set of windows of the training trajectories. Since every window starts from an observed state, the learned field is integrated over a time $h\Delta t_{\mathrm{obs}}=1$ at most, and the error does not accumulate along the whole trajectory. The penalized model minimizes \eqref{eq:penalty_loss_exp} with $\mathcal L$ replaced by $\mathcal L_{\mathrm{traj}}$.

\paragraph{Integration and gradients.} During training, $F_\theta$ is integrated in single precision with RK4 and fixed step $\delta t=0.0125$, i.e.\ $8$ steps per observation interval, and the gradient of \eqref{eq:traj_loss} is computed by backpropagation through the steps of the solver. In an explicit Runge-Kutta scheme every stage vanishes at a zero of the field, so that the planted equilibria are exact fixed points also of the discrete flow used in training, for any step. The step is chosen small enough for the discretization error not to exceed the accuracy of the planted equilibria. Otherwise, their exactness would bring no benefit to the fit of the trajectories. On the training windows, RK4 with step $\delta t$ applied to $F_{\mathrm{base}}$ differs from the reference solution by at most $2.1\times10^{-6}$, of the same order as the fixed-point residual of the QR-constrained model.

\paragraph{Training protocol.} The models are trained with the protocol of the study of the dependence on the dataset size in Section~\ref{subsec:Lotka-Volterra_experiment}, with the loss \eqref{eq:traj_loss} in place of the MSE on the field. The mini-batches consist of $128$ windows drawn at random, and the validation MSE is the loss \eqref{eq:traj_loss} on the windows of the validation trajectories. The number of optimizer steps is capped at $2\times10^6$. The penalty weight is selected as in that study, which gives $\lambda=10^{-1}$. All the $30$ reported runs ($10$ seeds for each model) stop by the learning-rate criterion and not by the maximum number of steps.

\paragraph{Evaluation.} We integrate the learned fields, with the integrator used in training, and $F_{\mathrm{base}}$, with the reference integrator, up to $t=30$ from $500$ test initial conditions, drawn uniformly in $[0.2,3.5]^2$ and not used for training or validation (the same as in Figure~\ref{fig:lv_trajectory_error}). Since the training trajectories end at $T=10$, we measure the trajectory error separately for $t\le T$ and for $t>T$: for each initial condition we compute the root mean square of $\|x_\theta(t)-x(t)\|_2$ over the observation times in $(0,10]$ and in $(10,30]$, and then the geometric mean over the initial conditions.

\begin{figure}[t]
    \centering
    \includegraphics[width=0.8\linewidth]{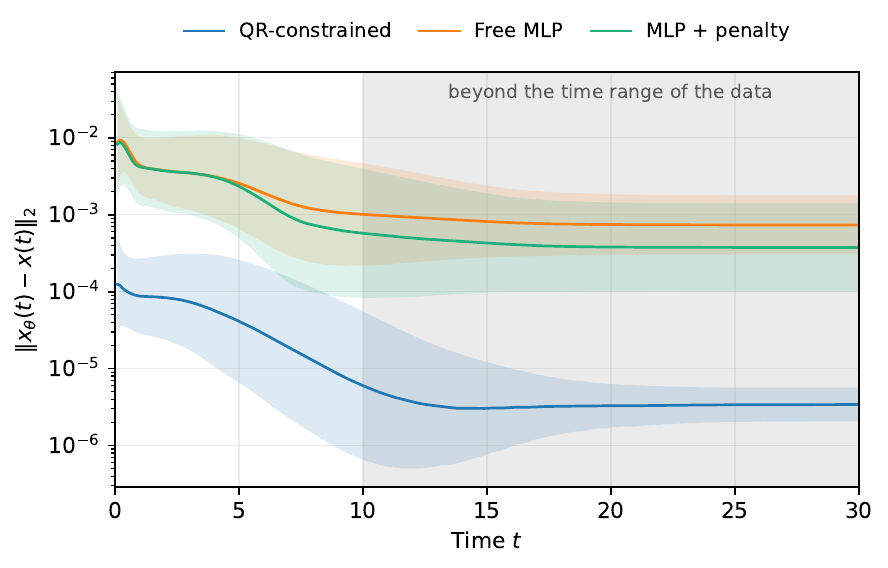}
    \caption{Error $\|x_\theta(t)-x(t)\|_2$ along the test trajectories, for the models trained on trajectories observed up to $T=10$. The learned fields and $F_{\mathrm{base}}$ are integrated from the same $500$ test initial conditions as in Figure~\ref{fig:lv_trajectory_error}, which are not used for training or validation. The shaded region marks the times $t>T$, longer than the duration of the training trajectories. The penalized MLP uses $\lambda=10^{-1}$. Lines and bands as in Figure~\ref{fig:lv_grid_rmse}, with the statistics taken over the initial conditions and the $10$ seeds.}
    \label{fig:traj_training_error}
\end{figure}

\begin{table}[t]
    \centering
    \footnotesize
    \setlength{\tabcolsep}{4pt}
    \caption{Results of the models trained on trajectories at the retained checkpoint, as geometric mean ${\times}/{\div}$ geometric standard deviation over $10$ seeds. The trajectory errors are root-mean-square errors along the test trajectories of Figure~\ref{fig:traj_training_error}, for $t\le T$ and for $t>T$, averaged geometrically over the initial conditions. The fixed-point residual $r_{\mathrm{fp}}$ is defined in \eqref{eq:fp_residual_exp}, and the Jacobian error is $\max_{l}\|DF_\theta(\bar x^{(l)})-DF_\star(\bar x^{(l)})\|_2$, as in Table~\ref{tab:lv_dataset_size}. In bold, the lowest geometric mean of each column.}
    \label{tab:traj_training}
    \resizebox{\linewidth}{!}{%
    \begin{tabular}{lcccc}
        \toprule
        Method & Traj.\ error, $0<t\le10$ & Traj.\ error, $10<t\le30$ & $r_{\mathrm{fp}}$ & Jacobian error \\
        \midrule
        QR-constrained & {\boldmath$1.0\times10^{-4}\,{\times}/{\div}\,1.36$} & {\boldmath$5.6\times10^{-6}\,{\times}/{\div}\,1.28$} & {\boldmath$2.7\times10^{-6}\,{\times}/{\div}\,1.76$} & {\boldmath$1.2\times10^{-1}\,{\times}/{\div}\,2.36$} \\
        Free MLP & $5.7\times10^{-3}\,{\times}/{\div}\,1.28$ & $9.4\times10^{-4}\,{\times}/{\div}\,1.80$ & $1.3\times10^{0}\,{\times}/{\div}\,2.41$ & $4.9\times10^{0}\,{\times}/{\div}\,1.34$ \\
        MLP + penalty & $5.4\times10^{-3}\,{\times}/{\div}\,2.02$ & $5.1\times10^{-4}\,{\times}/{\div}\,3.26$ & $6.3\times10^{-4}\,{\times}/{\div}\,2.87$ & $4.6\times10^{0}\,{\times}/{\div}\,1.39$ \\
        \bottomrule
    \end{tabular}}
\end{table}

\begin{table}[tb]
    \centering
    \footnotesize
    \setlength{\tabcolsep}{3pt}
    \caption{Results at each equilibrium for the models of Table~\ref{tab:traj_training} (QR: QR-constrained model; Free: free MLP; Pen.: MLP + penalty with $\lambda=10^{-1}$). The column Obs.\ gives the fraction of the training observations within distance $0.25$ of the equilibrium, averaged over the seeds. The residual $\|F_\theta(\bar x^{(l)})\|_2$, evaluated in double precision, and the Jacobian error $\|DF_\theta(\bar x^{(l)})-DF_\star(\bar x^{(l)})\|_2$ are geometric means over the $10$ seeds. The column Type gives the number of seeds for which the eigenvalues of $DF_\theta(\bar x^{(l)})$ have the same signs as those of $DF_\star(\bar x^{(l)})$.}
    \label{tab:traj_training_equilibria}
    \resizebox{\linewidth}{!}{%
    \begin{tabular}{llcccccccccc}
        \toprule
        & & & \multicolumn{3}{c}{$\|F_\theta(\bar x^{(l)})\|_2$} & \multicolumn{3}{c}{Jacobian error} & \multicolumn{3}{c}{Type} \\
        \cmidrule(lr){4-6}\cmidrule(lr){7-9}\cmidrule(lr){10-12}
        Equilibrium & & Obs. & QR & Free & Pen. & QR & Free & Pen. & QR & Free & Pen. \\
        \midrule
        $(0,0)$ & unstable node & $0.02\%$ & $1.7\times10^{-6}$ & $1.3\times10^{0}$ & $1.5\times10^{-4}$ & $1.1\times10^{-1}$ & $4.9\times10^{0}$ & $4.4\times10^{0}$ & 10/10 & 0/10 & 1/10 \\
        $(0,2)$ & stable node & $42.28\%$ & $1.9\times10^{-6}$ & $8.9\times10^{-4}$ & $4.2\times10^{-4}$ & $2.7\times10^{-3}$ & $1.1\times10^{-1}$ & $8.9\times10^{-2}$ & 10/10 & 10/10 & 10/10 \\
        $(3,0)$ & stable node & $20.16\%$ & $1.7\times10^{-6}$ & $1.4\times10^{-3}$ & $3.7\times10^{-4}$ & $1.2\times10^{-2}$ & $1.2\times10^{0}$ & $1.0\times10^{0}$ & 10/10 & 10/10 & 10/10 \\
        $(1,1)$ & saddle & $3.55\%$ & $1.8\times10^{-6}$ & $3.6\times10^{-3}$ & $4.5\times10^{-4}$ & $8.6\times10^{-4}$ & $5.0\times10^{-2}$ & $1.0\times10^{-1}$ & 10/10 & 10/10 & 10/10 \\
        \bottomrule
    \end{tabular}}
\end{table}

\paragraph{Results.} On the test trajectories, the QR-constrained model has the lowest trajectory error both for $t\le T$ and for $t>T$, for each seed, and the gap widens beyond $T$ (Figure~\ref{fig:traj_training_error} and Table~\ref{tab:traj_training}). The errors at $t=100$ coincide with those at $t=30$: as in the regression setting (Figure~\ref{fig:lv_trajectory_error}), the trajectories of the baselines converge to equilibria that are displaced from the true ones. Those of the QR-constrained model converge instead to the planted equilibria, which are exact also for the integrator. In addition, for the baselines up to four of the $500$ trajectories per seed end, at $t=30$, closest to an equilibrium different from that of the true trajectory, which never happens for the QR-constrained model.

The unstable node at the origin is almost absent from the training data (Table~\ref{tab:traj_training_equilibria}), since the trajectories move away from it. The free MLP has no equilibrium close to the origin in any seed. Newton's method started at the origin converges to a learned equilibrium at distance $1.4$ or more, which accounts for its fixed-point residual of order one in Table~\ref{tab:traj_training}. The penalized MLP places a zero close to the origin, but of the wrong type in $9$ of the $10$ seeds. At the three equilibria visited by the trajectories, instead, all the models have the correct type, and the QR-constrained model has the lowest Jacobian error. Finally, the penalty lowers the fixed-point residual of the free MLP by orders of magnitude, but changes its trajectory errors by less than a factor of two.

\end{document}